\documentclass[12pt]{article}

\renewcommand{\title}[1]{
\begin{center} \Large \bf #1 \end{center}
}

\renewcommand{\author}[2]{
 \begin{center} #1  \vspace{3mm} \\
  #2 \\
 \end{center}
\addvspace{\baselineskip}
}

\usepackage{amssymb}
\usepackage{amsmath}

\usepackage{amsthm}
\newtheorem{thm}{Theorem}[section]
\newtheorem{prop}[thm]{Proposition}
\newtheorem{cor}[thm]{Corollary}
\newtheorem{lem}[thm]{Lemma}

\theoremstyle{definition}
\newtheorem{defn}{Definition}

\theoremstyle{remark}
\newtheorem*{rem}{Remark}

\makeatletter
\@addtoreset{equation}{section}

\makeatother

\begin{document}

\baselineskip 5mm

\title{Twisted Fock Representations of Noncommutative K\"ahler Manifolds}

\author{${}^1$ Akifumi Sako and~ ${}^2$
Hiroshi Umetsu }{
${}^1$  Department of Mathematics,
Faculty of Science Division II,\\
Tokyo University of Science,
1-3 Kagurazaka, Shinjuku-ku, Tokyo 162-8601, Japan\\
${}^1$  
Fakult\"at f\"ur Physik, Universit\"at Wien\\
Boltzmanngasse 5, A-1090 Wien, Austria\\
${}^2$
National Institute of Technology, Kushiro College\\
Otanoshike-nishi 2-32-1, Kushiro, Hokkaido 084-0916, Japan}

\noindent
{\bf MSC 2010:} 53D55 , 81R60 
\vspace{1cm}

\abstract{We introduce twisted Fock representations of
noncommutative K\"ahler manifolds and give their explicit
expressions. 
The twisted Fock representation is a representation of the
Heisenberg like algebra whose states are constructed by acting creation
operators on a vacuum state. ``Twisted" means that creation
operators are not Hermitian conjugate of annihilation operators in
this representation. In deformation quantization of K\"ahler manifolds
with separation of variables formulated by Karabegov, local complex
coordinates and partial derivatives of the K\"ahler potential with
respect to coordinates satisfy the commutation relations between
the creation and annihilation operators. Based on these relations, we
construct the twisted Fock representation of noncommutative K\"ahler
manifolds and give a dictionary to translate between the twisted Fock
representations and functions on noncommutative K\"ahler manifolds
concretely.  }

\section{Introduction}

Deformation quantization is a way to construct
noncommutative geometry, which is first
introduced by Bayen, Flato, Fronsdal,
Lichnerowicz and Sternheimer \cite{Bayen:1977ha}. 
Several ways of deformation
quantization were established by \cite{DeW-Lec, Omori, Fedosov, Kontsevich}.
In particular, deformation quantizations of K\"ahler manifolds were
provided in \cite{Moreno86a, Moreno86b, Cahen93, Cahen95}.
In this article, the deformation
quantization with separation of variables is used to construct
noncommutative K\"ahler manifolds that is introduced by Karabegov
\cite{Karabegov, Karabegov1996, Karabegov2011}.
 (For a recent review, see \cite{Schlichenmaier}.)
The deformation quantization is an associative algebra on 
a set of formal power series of $C^{\infty}$ functions with
a star product between formal power series.
One of the advantages of deformation quantization is that
usual analytical techniques are available on noncommutative
manifolds constructed in this way.
On the other hand, when we consider field theories on noncommutative 
manifolds given by deformation quantization, 
physical quantities are given as formal power series,
and there are difficulties to understand them from a viewpoint of physics.
A typical way to solve the difficulties is to make a representation 
of the noncommutative algebra.

The purpose of this article is to construct
the Fock representation of noncommutative K\"ahler manifolds.
The algebras on noncommutative K\"ahler manifolds which are constructed
by deformation quantization with separation of variables contain the
Heisenberg like algebras. Local complex coordinates and partial
derivatives of a K\"ahler potential satisfy the commutation relations
between creation and annihilation operators.  A Fock space is spanned by
a vacuum, which is annihilated by all annihilation operators, and states
obtained by acting creation operators on this vacuum.
The algebras on noncommutative K\"ahler manifolds are represented as
those of linear operators acting on the Fock space. 
We call the representation of the algebra the Fock representation.
In representations studied in this article, creation operators and
annihilation operators are not Hermitian conjugate with each other, in general.
Therefore, the bases of the Fock space are not the Hermitian conjugates of
those of the dual vector space.
In this case, we call the representation the twisted Fock representation.
Historically, Berezin constructed a kind of the Fock representations of some
noncommutative K\"ahler manifolds\cite{Berezin:1974,Berezin:1975wd}, 
and since then there have been
various works on this subject 
\cite{Rawnsley,Schlichenmaier,Schlichenmaier1,Perelomov}.  
In this article, we construct the twisted Fock representation for
an arbitrary noncommutative K\"ahler manifold given by 
deformation quantization with separation of variables 
\cite{Karabegov, Karabegov1996, Karabegov2011}.


One of the main results in this article is summarized as the following 
dictionary, Table \ref{table:result}.
\begin{table}[hbtp]
 \caption{Functions - Fock operators Dictionary }
 \label{table:result}
 \centering
 \begin{tabular}{c|c}
  \hline
  Functions & Fock operators  \\
  \hline \hline
  $\displaystyle e^{-\Phi/\hbar}$  &  $|\vec{0}\rangle \langle \vec{0}|$
      \\
  \hline
  $z_i $    &   $a_i^{\dagger}$ 
      \\ \hline
  $\displaystyle \frac{1}{\hbar}\partial_i \Phi$   &  $\underline{a}_i$
      \\ \hline
  $\bar{z}^i$   
  & $\displaystyle a_i = 
      \sum \sqrt{\frac{\vec{m}!}{\vec{n}!}} H_{\vec{m},\vec{k}}
      H^{-1}_{\vec{k}+\vec{e}_i , \vec{n}} 
      |\vec{m}\rangle \underline{\langle \vec{n} |} $ 
      \\ \hline
  $\displaystyle \frac{1}{\hbar} \partial_{\bar i} \Phi$ 
  & 
      $\displaystyle \underline{a}_i^{\dagger} = \sum 
      \sqrt{\frac{\vec{m}!}{\vec{n}!}}
      (k_i+1) H_{\vec{m},\vec{k}+\vec{e}_i}
      H^{-1}_{\vec{k} , \vec{n}} |\vec{m}\rangle 
      \underline{ \langle \vec{n} |} $
      \\
  \hline
 \end{tabular}
\end{table}
In this dictionary, $z^i , \bar{z}^i \ (i=1, \cdots N)$ are local complex 
coordinates
of some open subset of an $N$ dimensional K\"ahler manifold. 
$\Phi$ is a K\"ahler potential and 
$H$ is defined by
$\displaystyle e^{\Phi/\hbar}= 
\sum H_{\vec{m},\vec{n}}z^{\vec{m}}\bar{z}^{\vec{n}}$, where
$z^{\vec{m}} = z_1^{m_1} z_2^{m_2} \cdots z_N^{m_N} $ for
$\vec{m}=(m_1, m_2, \cdots , m_N)$, and $\bar{z}^{\vec{n}}$ is similarly
defined.
$a_i^{\dagger}$ and $\underline{a}_i$ are essentially a creation operator 
and an annihilation operator, respectively.
$a_i$ and $\underline{a}_i^{\dagger}$ are 
Hermitian conjugate with each other.
Note that $a_i^{\dagger}$ is not a Hermitian conjugate of 
$\underline{a}_i$, in general.
More detailed definitions are given in Section \ref{sect2} and \ref{sect3}.

The twisted Fock algebra is defined on a local coordinate chart.
The star product with separation variables are glued between charts
with nonempty intersections.
Therefore, transition functions between the twisted Fock algebras 
on two charts having an overlapping region 
are also constructed.
Trace operations for the Fock representations as integrations of concerned 
functions are discussed.
We observe several examples, ${\mathbb C}^N$, a cylinder, ${\mathbb C}P^N$
 and ${\mathbb C}H^N$. 


The organization of this article is as follows.
In Section \ref{sect2}, we review several facts of 
deformation quantization with separation of variables 
which are used in this article.
In Section \ref{sect3}, a twisted Fock representation
is constructed on a chart of a general K\"ahler manifold.
In Section \ref{sect4}, transition maps between the twisted Fock
representations on two local coordinate charts are constructed.
In Section \ref{sect5}, we discuss a trace
operation for the twisted Fock representation.
In Section \ref{sect6}, the Fock representations of 
${\mathbb C}^N$, a cylinder, ${\mathbb C}P^N$
and  ${\mathbb C}H^N$ are given as examples.
We summarize our results in Section \ref{sect7}.

\section{A review of the deformation quantization with separation of
 variables} 
\label{sect2}

We give a general definition of deformation quantization, before
moving into the deformation quantization for K\"ahler manifolds.
\begin{defn}[Deformation quantization (weak sense)]
Let $M$ be a Poisson manifold.
$\cal F$ is defined as a set of formal power series:
\begin{eqnarray}
{\cal F} := \left\{  f \ \Big| \ 
f = \sum_k f_k \hbar^k, ~f_k \in C^\infty (M)
\right\} .
\end{eqnarray}
Deformation quantization is defined as a structure of
associative algebra of $\cal F$ whose product is defined by
a star product.
The star product is defined as 
\begin{eqnarray}
f * g = \sum_k C_k (f,g) \hbar^k
\end{eqnarray}
such that the product satisfies the following conditions.
\begin{enumerate}
\item $*$ is associative product.
\item $C_k$ is a bidifferential operator.
\item $C_0$ and $C_1$ are defined as 
\begin{eqnarray}
&& C_0 (f,g) = f g,  \\
&&C_1(f,g)-C_1(g,f) = i \{ f, g \}, \label{weakdeformation}
\end{eqnarray}
where $\{ f, g \}$ is the Poisson bracket.
\item $ f * 1 = 1* f = f$.
\end{enumerate}
\end{defn}
Note that this definition of deformation 
quantization is weaker than the usual definition
of deformation quantization.
The difference between them is in (\ref{weakdeformation}).
In the strong sense of deformation quantization the condition 
$C_1(f,g)= \frac{i}{2} \{ f, g \}$ 
is required.

As a special case of deformation quantizations of K\"ahler manifold $M$,
deformation quantization with separation of variables is 
introduced by Karabegov \cite{Karabegov, Karabegov1996, Karabegov2011}.
\begin{defn}[A star product with separation of variables]
$*$ is called a star product with separation of variables when 
\begin{eqnarray}
a * f = a f \label{l_holo}
\end{eqnarray}
for a holomorphic function $a$ and
\begin{eqnarray} 
f * b = f b \label{r_antiholo}
\end{eqnarray}
for an anti-holomorphic function $b$.
\end{defn}
The deformation quantization defined by using such a star product
is also denoted deformation quantization with separation of variables.
In this article, we consider only this type of deformation quantization
for K\"ahler manifolds.

Let $M$ be an $N$-dimensional complex K\"ahler manifold, $\Phi$ be its K{\" a}hler potential and $\omega$ be its K{\" a}hler 2-form:
\begin{eqnarray}
\omega &:=& i g_{k \bar{l}} dz^{k} \wedge d \bar{z}^{l} ,
\nonumber \\
g_{k \bar{l}} &:=& 
\frac{\partial^2 \Phi}{\partial z^{k} \partial \bar{z}^{l}} .
\end{eqnarray}
Here $g$ is the K\"ahler metric and 
$z^i , \bar{z}^j ~ (i,j = 1, \cdots , N)$ are
local coordinates on an open set $U \subset M$
which is diffeomorphic to a connected open subset 
of ${\mathbb C}^{N}$.
In this paper, we use the Einstein summation convention over repeated
indices. 
The $g^{\bar{k} l}$ is the inverse of the metric $g_{k \bar{l}}$:
\begin{eqnarray}
 g^{\bar{k} l}  g_{l \bar{m}} = \delta_{\bar{k} \bar{m} } .
\end{eqnarray}
In the following, we use the following abridged notations
\begin{align}
\partial_k = \frac{\partial}{\partial z^{k}} , \qquad
\partial_{\bar{k}} = \frac{\partial}{\partial \bar{z}^{k}}.
\end{align}

Karabegov constructed 
a star product with separation of variables for K\"ahler manifolds
in terms of differential operators  \cite{Karabegov,Karabegov1996},
as briefly explained  below.
For the left star multiplication by $f \in {\cal F}$, 
there exists  a differential
operator $L_f$ such that
\begin{equation}
 L_f g = f * g .
\end{equation}
$L_f$ is given as a formal power series in $\hbar$
\begin{equation}
 L_f = \sum_{n=0}^{\infty} \hbar^n A^{(n)},
  \label{Lf-An}
\end{equation}
where $A^{(n)}$ is a differential operator which contains only partial
derivatives  by $z^i ~(i=1, \cdots , N)$ and has the following form
\begin{equation}
 A^{(n)} = \sum_{k\geq 0} a^{(n;k)}_{\bar{i}_1 \cdots \bar{i}_k} 
  D^{\bar{i}_1} \cdots D^{\bar{i}_k},
  \label{An-a}
\end{equation}
where
\begin{equation}
 D^{\bar i} = g^{{\bar i} j} \partial_j ,
\end{equation}
and each $a^{(n;k)}_{\bar{i}_1 \cdots \bar{i}_k}$ is a $C^{\infty}$
function on $M$. 
In particular, $a^{(n;0)}$ acts as a multiplication operator.
Note that the differential operators $D^{\bar i}$ satisfy the following relations,
\begin{align}
 [ D^{\bar i} , D^{\bar j} ] &= 0, \\
 [D^{\bar{i}}, \partial_{\bar{j}}\Phi] &= \delta_{ij}.
 \label{D-dphi}
\end{align}

Karabegov showed the following theorem.
\begin{thm}[Karabegov\cite{Karabegov,Karabegov1996}]
$L_f$ is uniquely determined by requiring the following conditions,
\begin{align}
 L_f 1 = f * 1 =f, 
 \label{Lf-dphi1}\\
 [L_f , \partial_{\bar i} \Phi + \hbar \partial_{\bar i}] = 0,
 \label{Lf-dphi2} 
\end{align}
\end{thm}
This star product $*$ satisfies the associative condition
\begin{equation}
 h * (g * f) = (h * g) * f.
\end{equation}
%
Here is a useful theorem given by Karabegov.
\begin{thm}[Karabegov\cite{Karabegov,Karabegov1996}]
The differential operator $L_f$ for an arbitrary function $f$ is
obtained from the operator $L_{{\bar z}^i}$, which corresponds to
the left $*$ multiplication of ${\bar z}^i$,
\begin{equation} \label{Lf_Lz}
 L_f = \sum_{\alpha} \frac{1}{\alpha !} 
  \left(\frac{\partial}{\partial {\bar z}}\right)^\alpha 
  f (L_{\bar z} - {\bar z})^\alpha, 
\end{equation} 
where $\alpha$ is a multi-index.
\end{thm}

Similarly, the differential operator 
$\displaystyle{R_f = \sum_{n=0}^\infty \hbar^n B^{(n)}}$ corresponding to
the right $*$ multiplication by a function $f$ contains only partial
derivatives by $\bar{z}^i$ and is determined by the conditions
\begin{align} 
 & R_f 1 = 1*f = f, \label{Rf-dphi1} \\
 & [R_f, \partial_i \Phi + \hbar{\partial_i}] = 0.
 \label{Rf-dphi2} 
\end{align}
$B^{(n)}$ has the following form,
\begin{align}
  B^{(n)} = \sum_{k\geq 0} b^{(n;k)}_{i_1 \cdots i_k} 
 D^{i_1} \cdots D^{i_k},
 \label{Br}
\end{align}
where $D^i = g^{i\bar{j}}\partial_{\bar{j}}$ and 
$b^{(n;k)}_{i_1 \cdots i_k} \in C^{\infty}(M)$.
The differential operator $R_f$ for an arbitrary function $f$ is
obtained from the operator $R_{z^i}$, which corresponds to
the right $*$ multiplication by $z^i$,
\begin{equation} \label{Rf_Rz}
 R_f = \sum_{\alpha} \frac{1}{\alpha !} 
  \left(\frac{\partial}{\partial z}\right)^\alpha 
  f (R_z - z)^\alpha. 
\end{equation} 

In particular, the left star product by $\partial_i \Phi$ and the right star
product by $\partial_{\bar{i}} \Phi$ are respectively written as
\begin{align}
 L_{\partial_i \Phi} &= \hbar \partial_i + \partial_i \Phi
 = \hbar e^{-\Phi/\hbar} \partial_i e^{\Phi/\hbar}, 
 \label{LdPhi}\\
 R_{\partial_{\bar{i}} \Phi} 
 &= \hbar \partial_{\bar{i}} + \partial_{\bar{i}} \Phi
 = \hbar e^{-\Phi/\hbar} \partial_{\bar{i}} e^{\Phi/\hbar}.
 \label{RdPhi}
\end{align} 
{}{}From the definition of the star product,
we easily find 
\begin{align}
[ \frac{1}{\hbar} 
\partial_i \Phi ,~ z^j ]_*  &= \delta_{ij}, 
 \qquad 
 [z^i ,~ z^j ]_* = 0, \qquad
 [ \partial_i \Phi ,~ \partial_j \Phi ]_* 
  = 0, 
 \label{comm-rel-1} \\
 [ \bar{z}^i ,~  \frac{1}{\hbar} \partial_{\bar{j}}\Phi ]_* 
 &= \delta_{ij}, \qquad
 [ \bar{z}^i ,~ \bar{z}^j ]_* = 0, \qquad
 [\partial_{\bar{i}} \Phi ,~ \partial_{\bar{j}} \Phi]_* 
 = 0,
 \label{comm-rel-2}
\end{align}
where $[ A ,~ B]_*= A* B - B*A$. 
Hence, $\{z^i, \partial_j \Phi ~|~ i, j=1, 2, \cdots, N \}$ and 
$\{\bar{z}^i, \partial_{\bar{j}} \Phi ~|~ i, j=1, 2, \cdots, N \}$ 
constitute $2N$ sets of the creation and annihilation operators
under the star product. But, it should be noted that operators in 
$\{z^i, \partial_j \Phi\}$ does not commute with ones in 
$\{\bar{z}^i, \partial_{\bar{j}} \Phi\}$, e.g., 
$z^i * \bar{z}^j - \bar{z}^j * z^i \neq 0$.


\section{The Fock representation of noncommutative K\"ahler manifolds}
\label{sect3}
In this section we introduce the Fock space on an 
open set $U \subset M$ which is diffeomorphic to a connected open subset 
of ${\mathbb C}^{N}$ and an algebra as a set of linear operators acting on 
the Fock space.

As mentioned in Section \ref{sect2}, 
from the (\ref{comm-rel-1}) and (\ref{comm-rel-2}) 
$\{z^i, \partial_j \Phi ~|~ i, j=1, 2, \cdots, N \}$ and 
$\{\bar{z}^i, \partial_{\bar{j}} \Phi ~|~ i, j=1, 2, \cdots, N \}$ 
are candidates for the creation and annihilation operators
under the star product $*$.
We introduce 
$a^\dagger_i, a_i, \underline{a}^{\dagger}_i$ and $\underline{a}_i 
~(i= 1,2, \dots , N)$ by
\begin{align}
 a^\dagger_i =  z^i, ~~~ 
 \underline{a}_i = \frac{1}{\hbar}  \partial_i \Phi,~~~
 a_i =\bar{z}^i, ~~~  
 \underline{a}_i^\dagger = \frac{1}{\hbar}  \partial_{\bar{i}} \Phi.
\label{creation_anihilation}
\end{align} 
Then they satisfy the following commutation relations which are similar
to the usual commutation relations for the creation and annihilation operators but slightly
different, 
\begin{align}
[ \underline{a}_i ,~ a^{\dagger}_j ]_*  &= \delta_{ij}, 
 \qquad 
 [a_i^{\dagger} ,~ a_j^{\dagger} ]_* = 0, \qquad
 [ \underline{a}_i ,~ \underline{a}_j  ]_* 
  = 0, 
 \label{comm-rel-1+1} \\
 [ a_i,~  \underline{a}_j^\dagger ]_* 
 &= \delta_{ij}, \qquad
 [ \underline{a}_i^{\dagger} ,~ \underline{a}_j^{\dagger} ]_* = 0, 
 \qquad
 [a_i  ,~ a_j ]_* 
 = 0.
 \label{comm-rel-2+1}
\end{align}
There are differences from ordinary creation and annihilation operators 
that these two sets of creation and annihilation operators 
are not given as direct sum, in other words, 
\begin{align}
[ a_i , a_i^{\dagger} ]_*  ~~  \mbox{and} ~~  
[ \underline{a}_i  , ~ \underline{a}_j^{\dagger} ]_*
\end{align}
do not vanish in general.

The star product with separation of variables has the following property
under the complex conjugation.

\begin{prop}
\begin{align}
\overline{f*g}  = \overline{ L_f g } = \bar{g} * \bar{f}
\end{align}
 \label{comp-conj}
\end{prop}
\begin{proof}
 As described in the previous section, $L_f$ and $R_f$ are uniquely
 determined by the equations (\ref{Lf-dphi1}), (\ref{Lf-dphi2}),
 (\ref{Rf-dphi1}), and (\ref{Rf-dphi2}). 
 From the complex conjugation of
 (\ref{Lf-dphi1}) and (\ref{Lf-dphi2}), we find
 \begin{equation}
  \overline{L_f} 1 = \bar{f}, \qquad
  [\overline{L_f}, ~\partial_i \Phi + \hbar \partial_i] = 0.
 \end{equation} 
 Because of the uniqueness of solution of (\ref{Rf-dphi1}) and
 (\ref{Rf-dphi2}), $\overline{L_f}$ is equal to $R_{\bar{f}}$.
\begin{align}
 \overline{f*g} = \overline{L_f g} = \overline{L_f} \bar{g}
 = R_{\bar{f}} \bar{g} = \bar{g} * \bar{f}.
\end{align}
\end{proof}
%

The Fock space is defined by a vector space spanned by the bases
which is generated by acting 
$a_i^\dagger$ on $|\vec{0}\rangle$,
\begin{align}
 |{\vec n} \rangle  
 &= |n_1, \cdots, n_N \rangle  \nonumber \\
 &= c_1(\vec{n})
 (a^{\dag}_1)_{*}^{n_1} * \cdots * (a^{\dag}_N)_{*}^{n_N} * 
 |\vec{0} \rangle ,
 \label{base-n}
\end{align}
where $|{\vec 0} \rangle =  |0, \cdots, 0 \rangle $ satisfies
\begin{align}
\underline{a}_i * |{\vec 0} \rangle = 0  \ \ ( i= 1, \cdots , N ) 
\end{align}
and $\displaystyle (A)_*^n $ stands for $\overbrace{A * \cdots  * A}^n$.
$c_1(\vec{n})$ is a normalization coefficient which
does not depend on $z^i$ and $\bar{z}^i$.  
Here, we define the basis of a dual vector space by acting
$\underline{a}_i$ on $\langle \vec{0}|$,
\begin{align}
 \underline{\langle{\vec m}|} 
 &=  \underline{\langle m_1, \cdots, m_N|} \nonumber \\
 &=  \langle \vec{0}| *
 (\underline{a}_1)_{*}^{m_1} * \cdots * 
 (\underline{a}_N)_{*}^{m_N} 
 c_2(\vec{m}),
 \label{base-m}
\end{align}
and
\begin{align}
 \langle \vec{0}| * a_i^{\dagger} =0 \qquad 
 ( i= 1, \cdots , N ),
\end{align}
where $c_2({\vec{m}})$ is also a normalization constant.
The underlines are attached to the bra vectors in
order to emphasize that $\underline{\langle{\vec m}|}$ is not Hermitian
conjugate to $| \vec{m} \rangle$. 
In this article, we set the normalization constants as
\begin{align}
c_1 (\vec{n})= \frac{1}{\sqrt{\vec{n}!}},  \ \ \ 
c_2(\vec{n}) = \frac{1}{\sqrt{\vec{n}!}},
\end{align}
where $\vec{n}! = n_1 ! n_2 ! \cdots n_N !$.

\begin{defn}
The local twisted Fock algebra (representation) $F_U$ 
is defined as a algebra given by a set of linear operators
acting on the Fock space defined on $U$:
\begin{align}
F_U:= \{
 \sum_{ \vec{n} , \vec{m} } A_{\vec{n} \vec{m} } 	
|\vec{n} \rangle \underline{ \langle \vec{m} | }~ | 
 ~ A_{\vec{n} \vec{m} } \in {\mathbb C} \}.
\end{align}
and 
products between its elements are given by the star product $*$.
\end{defn}
In the remaining part of this section, 
we construct concrete expressions of functions which are elements of 
this local twisted Fock algebra.

\begin{lem} [Berezin] \label{lem-1}
For arbitrary K\"ahler manifolds $(M, \omega)$,
there exists a K\"ahler potential 
$\Phi (z^1,\dots ,z^N , \bar{z}^1 , \dots , \bar{z}^N)$ such that 
\begin{align}
\Phi( 0,\dots ,0 , \bar{z}^1 , \dots , \bar{z}^N) = 0 , ~
\Phi (z^1,\dots ,z^N , 0 , \dots , 0)=0. 
\label{Phi_cond}
\end{align}
\end{lem}
This is easily shown as follow.
If a K\"ahler potential $\Phi$ satisfying 
$g_{i\bar{j}} = \partial_i \partial_{\bar j} \Phi$
does not satisfy (\ref{Phi_cond}),
then we redefine a new K\"ahler potential $\Phi'$ as
\begin{align}
&\Phi'(z^1,\dots ,z^N , \bar{z}^1 , \dots , \bar{z}^N) \notag \\
& := 
\Phi (z^1,\dots ,z^N , \bar{z}^1 , \dots , \bar{z}^N)
-\Phi( 0,\dots ,0 , \bar{z}^1 , \dots , \bar{z}^N) 
- \Phi (z^1,\dots ,z^N , 0 , \dots , 0) .
\end{align}
$\Phi (z^1,\dots ,z^N , 0 , \dots , 0) $ is a holomorphic function and
$\Phi( 0,\dots ,0 , \bar{z}^1 , \dots , \bar{z}^N) $ is an
anti-holomorphic function.
K\"ahler potentials have ambiguities of adding holomorphic and
anti-holomorphic functions.
This $\Phi'$ satisfies the condition (\ref{Phi_cond}).
In the following, we abbreviate 
$\Phi (z^1,\dots ,z^N , \bar{z}^1 , \dots , \bar{z}^N)$ to 
$\Phi(z,\bar{z})$ for convenience.

In \cite{Sako:2012ws}, it is shown that $e^{-\Phi/\hbar}$ corresponds to
a vacuum projection operator $ |\vec{0} \rangle \langle \vec{0}| $ for
the noncommutative ${\mathbb C}P^N$.  We extend this statement for
general K\"ahler manifolds.
\begin{prop}
Let $(M, \omega)$ be a K\"ahler manifold, $\Phi$ be its K\"ahler
 potential with the property (\ref{Phi_cond}), and
$*$ be a star product with separation of variables given
in the previous section.
Then the following function
\begin{equation}
 |\vec{0} \rangle \langle \vec{0}| := e^{-\Phi/\hbar},
\end{equation}
satisfies
\begin{align}
 & \underline{a}_i * |\vec{0} \rangle \langle \vec{0}| 
 = 0, \qquad
 |\vec{0} \rangle \langle \vec{0}| * a_i^\dagger 
= 0, \\
 & \left(|\vec{0} \rangle \langle \vec{0}|\right) * 
 \left(|\vec{0} \rangle \langle \vec{0}|\right) =
 e^{-\Phi/\hbar}*e^{-\Phi/\hbar} = e^{-\Phi/\hbar}
 = |\vec{0} \rangle \langle \vec{0}|.
\end{align}
\end{prop}

\begin{proof}	
We define the following normal ordered quantity,
\begin{align}
: e^{- \sum_i a_i^{\dagger} \underline{a}_i} : ~ 
 := {\prod_{i=1}^N}_* \sum_{n=0}^{\infty} 
 \frac{(-1)^n}{n!} (a_i^{\dagger})_*^n * (\underline{a}_i)^n_*.
\end{align}
Here ${\prod_{i=1}^N}_* $ is defined by
${\prod_{i=1}^N}_* f_i = f_1 * f_2 * \cdots * f_N$.
Note that, if $i \neq j$, $\underline{a}_i$ commutes with 
$\underline{a}_j , a_j^{\dagger}$, and $a_i^{\dagger}$ commutes with 
${\underline{a}}_j , a_j^{\dagger}$.
Therefore $: e^{- \sum_i a_i^{\dagger} \underline{a}_i} :$ 
does not depend on the order of each factor
$\sum_{n=0}^{\infty} \frac{(-1)^n}{n!} 
(a_i^{\dagger})_*^n * (\underline{a}_i)^n_*$.

It is easy to show that
$ \underline{a}_i * \sum_{n=0}^{\infty} \frac{(-1)^n}{n!} 
(a_i^{\dagger})_*^n * (\underline{a}_i)^n_* = 0$,
in the same way as in the case of the ordinary harmonic oscillator,
\begin{align}
\underline{a}_i * \sum_{n=0}^{\infty} \frac{(-1)^n}{n!} 
(a_i^{\dagger})_*^n * (\underline{a}_i)^n_* 
& =  \sum_{n=0}^{\infty} \frac{(-1)^n}{n!}  
 \left[ n (a_i^{\dagger})_*^{n-1} (\underline{a}_i)_*^n 
 + (a^{\dagger})_*^n (\underline{a})_*^{n+1}
 \right]
 = 0,
\end{align}
where the commutation relations (\ref{comm-rel-1+1}) are used.
Similarly, we can show 
$ \sum_{n=0}^{\infty} \frac{(-1)^n}{n!} 
(a_i^{\dagger})_*^n * (\underline{a}_i)^n_* * a_i^\dagger = 0 $. 
These results and the fact that $\underline{a}_i$ and $a_i^{\dagger}$
commute with $\underline{a}_j$ and $a_j^{\dagger}$ for $ i\neq j$ lead to 
$ \underline{a}_i * : e^{-  \sum_i a_i^{\dagger} \underline{a}_i} : = 0$
 and  
 $ : e^{- \sum_i a_i^{\dagger} \underline{a}_i} : * a_i^\dagger = 0 $. 
 Further, these relations imply 
$: e^{- \sum_i a_i^{\dagger} \underline{a}_i} : * 
: e^{- \sum_i a_i^{\dagger} \underline{a}_i} :
 ~=~ : e^{- \sum_i a_i^{\dagger} \underline{a}_i} :$.

Therefore, all we have to do is to show
\begin{align}
: e^{-  \sum_i a_i^{\dagger} \underline{a}_i} : = e^{-\Phi/\hbar}.
\end{align}
This can be done as follows:
\begin{align}
: e^{- \sum_i a_i^{\dagger} \underline{a}_i} : 
 &= \sum_{\vec{n}} \frac{(-1)^{|n|}}{\vec{n}!} 
 (a^\dagger)_*^{\vec{n}} * (\underline{a})_*^{\vec{n}} \nonumber \\
 &= \sum_{\vec{n}} \frac{(-1)^{|n|}}{\vec{n}! \hbar^{|n|}} 
 (z)_*^{\vec{n}} * (\partial \Phi)_*^{\vec{n}}.
\label{3_16}
\end{align}
In this paper, we use the following notation: for an $N$-tuple 
$A_i ~(i=1, 2, \cdots, N)$ and an $N$-vector 
$\vec{n}=(n_1, n_2, \cdots, n_N)$,
\begin{align}
(A)_*^{\vec{n}} &= (A_1)_*^{n_1} * (A_2)_*^{n_2} * \cdots *
 (A_N)_*^{n_N}, \\ 
 \vec{n}! &= n_1! n_2! \cdots n_N!, \qquad
  |n| = \sum_{i=1}^N n_i.
\end{align}
By using $(z)^{\vec{n}}_* = (z)^{\vec{n}} = (z^1)^{n_1} \cdots (z^N)^{n_N}$ and 
(\ref{RdPhi}), (\ref{3_16}) is recast as
\begin{align}
& \sum_{n_1, n_2, \dots, n_N=0}^\infty 
\frac{1}{n_1 ! n_2 ! \cdots n_N! } 
(-z^1)^{n_1} \cdots (-z^N)^{n_N} e^{-\frac{\Phi (z, \bar{z})}{\hbar}}
 \partial_1^{n_1}  \cdots \partial_N^{n_N} 
e^{\frac{\Phi (z, \bar{z})}{\hbar}} \notag \\
&= e^{-\frac{\Phi (z, \bar{z})}{\hbar}} e^{\frac{\Phi (0 , \bar{z})}{\hbar}}
\notag \\
&=  e^{-\frac{\Phi (z, \bar{z})}{\hbar}} . \label{3_17}
\end{align}
Here, the final equality follows from the condition (\ref{Phi_cond}).

\end{proof}

From a similar calculation to the above proof, we can also show the
following relations with respect to $a_i$ and $\underline{a}_i^\dagger$, 
\begin{align}
 & |\vec{0} \rangle \langle \vec{0}| 
 = e^{-\Phi/\hbar} 
 = : e^{-\sum_i \underline{a}_i^\dagger a_i} :
 = {\prod_{i=1}^N}_* \sum_{n=0}^{\infty} 
 \frac{(-1)^n}{n!} (\underline{a}_i^{\dagger})_*^n * (a_i)^n_*, 
 \label{vac-proj1} \\
 & a_i * |\vec{0} \rangle \langle \vec{0}| = 0, \qquad
 |\vec{0} \rangle \langle \vec{0}| * \underline{a}_i^\dagger = 0.
 \label{vac-proj2}
\end{align}

\begin{lem}[Sako, Suzuki, Umetsu \cite{Sako:2012ws}] \label{ssu}
$e^{-\Phi/\hbar}=|0 \rangle \langle  0 |$ satisfies the relation 
\begin{align}
|0 \rangle \langle  0 |* f(z, \bar{z}) 
 &= e^{-\Phi/\hbar} * f(z, \bar{z}) 
 = e^{-\Phi/\hbar} f(0, \bar{z}) 
 = |0 \rangle \langle  0 | f(0, \bar{z}),
 \label{vacuum1} \\
 f(z, \bar{z})*|0 \rangle \langle  0 |
 &= f(z, \bar{z}) * e^{-\Phi/\hbar} 
 = f(z, 0) e^{-\Phi/\hbar}
 = f(z, 0) |0 \rangle \langle  0 |.
  \label{vacuum2}
\end{align}
for a function $f(z, \bar{z})$ such that $f(z, \bar{w})$ can be expanded
as Taylor series with respect to $z^i$ and $\bar{w}^j$, respectively. 
\end{lem}
This proof is given in \cite{Sako:2012ws}, but for the convenience 
its proof is reviewed here. 
\begin{proof}
To show the relation (\ref{vacuum1}), we note that the differential
operator $R_{z^i}$ corresponding to the right product of $z^i$ contains
only partial derivatives by $\bar{z}^j$, and thus commutes with
 $z^k$. Moreover, $R_{z^i}$ annihilates $e^{-\Phi/\hbar}$,
 $R_{z^i} e^{-\Phi/\hbar} = e^{-\Phi/\hbar} * z^i
 = |\vec{0}\rangle \langle \vec{0} |* a_i^\dagger=0 $ .
 From these and (\ref{Rf_Rz}), the relation (\ref{vacuum1})
is shown as
\begin{align}
 e^{-\Phi/\hbar} * f(z, \bar{z}) &= R_f e^{-\Phi/\hbar} \nonumber \\
 &= \sum_{k_1, \dots, k_N=0}^\infty \frac{1}{k_1! \cdots k_N!}
 \partial_1^{k_1} \cdots \partial_N^{k_N} f(z, \bar{z})
 \prod_{i=1}^{N} \left(R_{z^i} - z^i \right)^{k_i} e^{-\Phi/\hbar}
 \nonumber \\
 &= \sum_{k_1, \dots, k_N=0}^\infty \frac{1}{k_1! \cdots k_N!}
 \partial_1^{k_1} \cdots \partial_N^{k_N} f(z, \bar{z})
 \prod_{i=1}^{N} \left(- z^i \right)^{k_i} e^{-\Phi/\hbar}
 \nonumber \\
 &= e^{-\Phi/\hbar} f(0, \bar{z}). \label{e-f}
\end{align}
Similarly, (\ref{vacuum2}) follows from (\ref{Lf_Lz}) and
 (\ref{vac-proj2}).

\end{proof}

%
%

We expand a function 
$\exp \Phi(z, \bar{z})/\hbar$ as a power series,
\begin{align}
 e^{\Phi(z, \bar{z})/\hbar}
 = \sum_{\vec{m},\vec{n}} H_{\vec{m}, \vec{n}}
 (z)^{\vec{m}} (\bar{z})^{\vec{n}}
 \label{H}
\end{align}
where $(z)^{\vec{n}} = (z^1)^{n_1} \cdots (z^N)^{n_N}$ and $(\bar{z})^{\vec{n}} = (\bar{z}^1)^{n_1} \cdots (\bar{z}^N)^{n_N}$.
Since $\exp \Phi/\hbar$ is real and satisfies (\ref{Phi_cond}),
the expansion coefficients $H_{\vec{m},\vec{n}}$ obey   
\begin{align}
 \bar{H}_{\vec{m},\vec{n}} &= H_{\vec{n}, \vec{m}}, \\
 H_{\vec{0}, \vec{n}} &= H_{\vec{n}, \vec{0}}
 = \delta_{\vec{n}, \vec{0}}.
\end{align}

Using this expansion, the following relations are obtained. 
\begin{prop}
 The right $*$-multiplication of
 $ (\underline{a})_*^{\vec{n}}=(\partial \Phi /\hbar )_*^{\vec{n}} $
 on $|\vec{0} \rangle \langle \vec{0} |$ is
 related to the right $*$-multiplication of
 $(a)_*^{\vec{n}} = (\bar{z})_*^{\vec{n}}$ on
 $|\vec{0} \rangle \langle \vec{0}| $ as follows,
\begin{align}
 |\vec{0} \rangle \langle  \vec{0} | * (\underline{a})_*^{\vec{n}}
 &= |\vec{0} \rangle \langle  \vec{0} |
 * \left(\frac{1}{\hbar}\partial \Phi \right)_*^{\vec{n}} \nonumber \\ 
 &= \vec{n}! \sum_{\vec{m}} H_{\vec{n}, \vec{m}}
 |\vec{0} \rangle \langle  \vec{0} | * (\bar{z})_*^{\vec{m}}
 = \vec{n}! \sum_{\vec{m}} H_{\vec{n}, \vec{m}}
 |\vec{0} \rangle \langle  \vec{0} | * (a)_*^{\vec{m}}.
 \label{vac-dPhi}
\end{align}
 Similarly, the following relation holds,
 \begin{align}
  (\underline{a}^\dagger)_*^{\vec{n}}
  * |\vec{0} \rangle \langle  \vec{0} | 
  &= \left(\frac{1}{\hbar}\bar{\partial} \Phi \right)_*^{\vec{n}}
  * |\vec{0} \rangle \langle  \vec{0} |
  \nonumber \\
  &= \vec{n}! \sum_{\vec{m}} H_{\vec{m}, \vec{n}} (z)^{\vec{m}}
  * |\vec{0} \rangle \langle  \vec{0} |
  = \vec{n}! \sum_{\vec{m}} H_{\vec{m}, \vec{n}}
  (a^\dagger)_*^{\vec{m}}
  * |\vec{0} \rangle \langle  \vec{0} |.
  \label{dbarPhi-vac}
 \end{align}
\end{prop}
\begin{proof}
By using (\ref{LdPhi}),
\begin{align}
 |\vec{0} \rangle \langle \vec{0} |
  * \left(\frac{1}{\hbar}\partial \Phi \right)^{\vec{n}}_*  
&= |\vec{0} \rangle \langle  \vec{0} | * 
\left( e^{-\Phi /\hbar } (\partial)^{\vec{n}} e^{\Phi /\hbar } \right)
\notag \\
&= |\vec{0} \rangle \langle \vec{0} | * 
 \left(e^{-\Phi /\hbar } \sum_{\vec{m}} \vec{n}!
 H_{\vec{n}, \vec{m}} (\bar{z})^{\vec{m}}
 + O(z)\right) .
\end{align}
From (\ref{vacuum1}) and Lemma \ref{lem-1}, this is rewritten as
\begin{align}
\left.
 |\vec{0} \rangle \langle \vec{0} | * 
 \left(e^{-\Phi /\hbar } \sum_{\vec{m}} \vec{n}!
 H_{\vec{n}, \vec{m}} (\bar{z})^{\vec{m}}\right) 
\right|_{(z,\bar{z})=(0, \bar{z})}
 &= |\vec{0} \rangle \langle \vec{0} | *
 \sum_{\vec{m}} \vec{n}! H_{\vec{n}, \vec{m}} (\bar{z})^{\vec{m}}
 \nonumber \\
 &= |\vec{0} \rangle \langle \vec{0} | *
 \sum_{\vec{m}} \vec{n}! H_{\vec{n}, \vec{m}} (a)_*^{\vec{m}}.
\end{align}
 The relation (\ref{dbarPhi-vac}) is the complex conjugation of
 (\ref{vac-dPhi}). 
\end{proof}

If there exists the inverse matrix $H^{-1}_{\vec{m}, \vec{n}}$, then the
following relations also holds,
\begin{cor}
\begin{align}
 |\vec{0} \rangle \langle \vec{0} | *  (a)_*^{\vec{n}} 
&= \sum_{\vec{m}} \frac{1}{\vec{m}!} H^{-1}_{\vec{n}, \vec{m}}
 |\vec{0} \rangle \langle \vec{0} | * (\underline{a})_*^{\vec{m}}, \\
(a^\dagger)_*^{\vec{n}} *  |\vec{0} \rangle \langle \vec{0} | 
&= \sum_{\vec{m}} \frac{1}{\vec{m}!} H^{-1}_{\vec{m}, \vec{n}}
 (\underline{a}^\dagger)^{\vec{m}} *
 |\vec{0} \rangle \langle \vec{0} |, 
\end{align}
where $H^{-1}_{\vec{n},\vec{m}}$ 
is the inverse matrix of the matrix $H_{\vec{n}, \vec{m}}$,
 $\sum_{\vec{k}} H_{\vec{m}, \vec{k}} H^{-1}_{\vec{k}, \vec{n}}
 = \delta_{\vec{m}, \vec{n}}$.
\end{cor}

We introduce bases of the Fock representation as follows,
\begin{equation}
 |\vec{m} \rangle \underline{\langle \vec{n}|}
 := \frac{1}{\sqrt{\vec{m}! \vec{n}!}}
 (a^\dagger)_*^{\vec{m}} * |\vec{0} \rangle \langle \vec{0}|
 * (\underline{a})_*^{\vec{n}}
 = \frac{1}{\sqrt{\vec{m}! \vec{n}!}}
 (z)_*^{\vec{m}} * e^{-\Phi/\hbar}
 * \left(\frac{1}{\hbar} \partial \Phi\right)_*^{\vec{n}}.
\end{equation}
By using (\ref{vac-dPhi}), the bases are also written as
\begin{align}
 |\vec{m} \rangle \underline{\langle \vec{n}|}
 &= \sqrt{\frac{\vec{n}!}{\vec{m}!}} 
 \sum_{\vec{k}} H_{\vec{n}, \vec{k}} 
 (z)_*^{\vec{m}} * e^{-\Phi/\hbar} * (\bar{z})_*^{\vec{k}} 
 \nonumber \\
 &= \sqrt{\frac{\vec{n}!}{\vec{m}!}} 
 \sum_{\vec{k}} H_{\vec{n}, \vec{k}} 
 (z)^{\vec{m}} (\bar{z})^{\vec{k}} e^{-\Phi/\hbar}. \label{base_H_ver}
\end{align} 
The completeness of the bases are formally shown as
\begin{align}
 \sum_{\vec{n}} |\vec{n} \rangle \underline{\langle \vec{n}|}
 &= \sum_{\vec{m}, \vec{n}} H_{\vec{n}, \vec{m}} 
 (z)^{\vec{n}} (\bar{z})^{\vec{m}} e^{-\Phi/\hbar} \nonumber \\
 &= e^{\Phi/\hbar} e^{-\Phi/\hbar} \nonumber \\
 &= 1.
\end{align}
The $*$-products between the bases are calculated as
\begin{align}
 |\vec{m} \rangle \underline{\langle \vec{n}|} *
 |\vec{k} \rangle \underline{\langle \vec{l}|}
 &= \frac{1}{\sqrt{\vec{m}!\vec{n}!\vec{k}!\vec{l}!}}
 (a^\dagger)_*^{\vec{m}} * |\vec{0} \rangle \langle \vec{0}|
 * (\underline{a})_*^{\vec{n}} *
(a^\dagger)_*^{\vec{k}} * |\vec{0} \rangle \langle \vec{0}|
 * (\underline{a})_*^{\vec{l}} \nonumber \\
 &= \delta_{\vec{n}, \vec{k}}
 |\vec{m} \rangle \underline{\langle \vec{l}|}.
\end{align}
The behavior of the bases under the complex conjugation is different
from usual,
\begin{align}
 \overline{|\vec{m} \rangle \underline{\langle \vec{n}|}}
 &= \sqrt{\frac{\vec{n}!}{\vec{m}!}} \sum_{\vec{k}}
 H_{\vec{k}, \vec{n}} (z)^{\vec{k}} (\bar{z})^{\vec{m}} e^{-\Phi/\hbar}
 \nonumber \\
 &= \sqrt{\frac{\vec{n}!}{\vec{m}!}} \sum_{\vec{k}, \vec{l}}
 \sqrt{\frac{\vec{k}!}{\vec{l}!}}
 H_{\vec{k}, \vec{n}} H^{-1}_{\vec{m}, \vec{l}}~
 |\vec{k} \rangle \underline{\langle \vec{l}|}.
\end{align}

The creation and annihilation operators $a_i^\dagger, \underline{a}_i$
act on the bases as follows,
\begin{align}
 a_i^\dagger * |\vec{m} \rangle \underline{\langle \vec{n}|}
 &= \sqrt{m_i+1} 
 |\vec{m} + \vec{e}_i \rangle \underline{\langle \vec{n}|}, \\
 \underline{a}_i * |\vec{m} \rangle \underline{\langle \vec{n}|}
 &= \sqrt{m_i} 
 |\vec{m} - \vec{e}_i \rangle \underline{\langle \vec{n}|}, \\
 |\vec{m} \rangle \underline{\langle \vec{n}|} * a_i^\dagger 
 &= \sqrt{n_i} 
 |\vec{m} \rangle \underline{\langle \vec{n} - \vec{e}_i|}, \\
 |\vec{m} \rangle \underline{\langle \vec{n}|} * \underline{a}_i
 &= \sqrt{n_i+1} 
 |\vec{m} \rangle \underline{\langle \vec{n} + \vec{e}_i|}, 
\end{align}
where $\vec{e}_i$ is a unit vector, $(\vec{e}_i)_j = \delta_{ij}$.
The action of $a_i$ and $\underline{a}_i^\dagger$ is derived by the
Hermitian conjugation of the above equations.

The creation and annihilation operators can be expanded with respect to
the bases,
\begin{align}
 a_i^\dagger &= \sum_{\vec{n}} \sqrt{n_i+1} 
 |\vec{n} + \vec{e}_i \rangle \underline{\langle \vec{n}|}, 
 \label{ca-op1} \\
 \underline{a}_i &= \sum_{\vec{n}} \sqrt{n_i+1}
 |\vec{n} \rangle \underline{\langle \vec{n} + \vec{e}_i|}, 
 \label{ca-op2} \\
 a_i &= \sum_{\vec{m}, \vec{n}, \vec{k}}
 \sqrt{\frac{\vec{m}!}{\vec{n}!}} H_{\vec{m}, \vec{k}}
 H^{-1}_{\vec{k} + \vec{e}_i, \vec{n}} 
 |\vec{m} \rangle \underline{\langle \vec{n}|}, \\
 \underline{a}_i^\dagger &=
 \sqrt{\frac{\vec{m}!}{\vec{n}!}} (k_i + 1) 
 H_{\vec{m}, \vec{k}+\vec{e}_i} H^{-1}_{\vec{k}, \vec{n}} 
 |\vec{m} \rangle \underline{\langle \vec{n}|}.
\end{align}

\section{Transition maps} \label{sect4}

Let $\{ U_a \} $ with $M= \cup_a U_a$ be a locally finite open covering and
$\{( U_a , \phi_a) \}$ be an atlas , 
where $\phi_a : U_a \rightarrow {\mathbb C}^N$.
Consider the case $U_{a} \cap U_{b} \neq \emptyset$.
Denote by $\phi_{a, b}$ the transition map from
$\phi_{a}(U_{a} ) $ to $\phi_{b}(U_{b} ) $.
The local coordinates 
$(z , \bar{z})=(z^1, \cdots ,z^N , \bar{z}^1 , \cdots , \bar{z}^N)$ 
on $U_{a}$ are transformed into the coordinates 
$(w , \bar{w})= (w^1, \cdots , w^N, \bar{w}^1 , \cdots , \bar{w}^N) $
on $U_{b}$ by $(w , \bar{w})= (w(z) , \bar{w}(\bar{z}) )$, where 
$w(z)=(w^1(z), \cdots , w^N(z))$
is a holomorphic function and 
$\bar{w}(\bar{z})= (\bar{w}^1 (\bar{z}) , \cdots , \bar{w}^N(\bar{z}))$ 
is an anti-holomorphic function.
Denote by $f*_a g$ and $f*_b g$ the star products defined 
in Section \ref{sect2} on $U_{a}$ and
$U_{b}$, respectively.
In general, there is a nontrivial transition maps $T$ between
two star products i.e. $f *_b g = T(f) *_a T(g)$.
But the transition maps are trivial in our case.

\begin{prop}
For an overlap $U_{a} \cap U_{b} \neq \emptyset$,
\begin{align}
f *_b g (w,\bar{w})= \phi^{*}_{a,b}~ f *_a  g  (w,\bar{w})
= \phi^{*}_{a,b}~ f(w(z) , \bar{w}(\bar{z})) *_a g(w(z) , \bar{w}(\bar{z})).
\end{align}
Here $\phi^{*}_{a,b}$ is the pull back of $\phi_{a,b}$.
\end{prop}
\begin{proof}
The K\"ahler potentials $\Phi_a (z,\bar{z})$ on $U_a$ and 
$\Phi_b (w,\bar{w})$ on $U_b$ satisfy, in general,
\begin{align*}
\Phi_b (w,\bar{w}) =  \Phi_a (z,\bar{z}) + \phi(z) + \bar{\phi}(\bar{z}),
\end{align*}
where $\phi$ is a holomorphic function and $\bar{\phi}$ is 
an anti-holomorphic function. 
We define a differential operator 
$L_{b ,f}$ by $L_{b ,f} g : = f*_b g$ on $U_b$.
Similarly, we use $g^{\bar{i}j}_b$ , $ D^{\bar{i}}_b$ etc. as 
the metric on $U_b$, differential operator $ D^{\bar{i}}$ on $U_b$, etc.
As mentioned in Section \ref{sect2},
\begin{eqnarray}
L_{b, f} = \sum_{n=0}^{\infty} \hbar^n a^b_{n, \vec{\bar{i}}}(f) 
D^{\vec{\bar{i}}}_b
= \sum_{n=0}^{\infty} \hbar^n \sum_{k\geq 0} 
a^{b(n;k)}_{\bar{i}_1 \cdots \bar{i}_k}  
D^{\bar{i}_1}_b \cdots D^{\bar{i}_k}_b,
\end{eqnarray}
is determined by 
\begin{align} \label{a_b_1}
[L_{b, f}~ ,~ R_{b,\partial_l \Phi_b} ]=
[L_{b, f}~ ,~ \frac{\partial \Phi_b}{\partial {\bar{w}^l}} 
 + \hbar  \frac{\partial}{\partial {\bar{w}^l}}]=0
\end{align}
On the  overlap 
$U_{a} \cap U_{b} $,
\begin{align}
D^{\bar{i}}_b= \frac{\partial \bar{w}^i}{\partial {\bar{z}^l}} D^{\bar{l}}_a ,
\end{align}
because 
$\displaystyle g^{\bar{i}j}_b = 
\frac{\partial \bar{w}^i}{\partial {\bar{z}^k}} 
\frac{\partial {w}^j}{\partial z^l} g^{\bar{k}l}_a $.
{}From the fact that differential operators $D^{\vec{\bar{i}}}_b$
contain only differentiation with respect to holomorphic coordinates $w^i$,
$D^{\vec{\bar{i}}}_b$ commutes with anti-holomorphic functions, 
then we obtain
\begin{align}
L_{b,f} = \sum_{n=0}^{\infty} \hbar^n a^b_{n, \vec{\bar{i}}}(f) 
{\left(\frac{\partial \bar{w}}{\partial {\bar{z}}}
 \right)^{\vec{\bar{i}}}}_{\vec{\bar{j}}} D^{\vec{\bar{j}}}_b,
\end{align}
where 
$\displaystyle{\left(\frac{\partial \bar{w}}{\partial {\bar{z}}}
 \right)^{\vec{\bar{i}}}}_{\vec{\bar{j}}}$ is an anti-holomorphic function
\begin{align}
{\left(\frac{\partial \bar{w}}{\partial {\bar{z}}}
 \right)^{\vec{\bar{i}}}}_{\vec{\bar{j}}}
= \frac{\partial \bar{w}^{i_1}}{\partial {\bar{z}^{j_1}}}
\cdots \frac{\partial \bar{w}^{i_k}}{\partial {\bar{z}^{j_k}}}.
\end{align}
Here, the Einstein summation convention over repeated
indices is also used for multi indices like $ \vec{\bar{i}} $ and so on.
Then 
\begin{align}
&\left[ L_{b, f}~ ,~ \frac{\partial \Phi_b}{\partial {\bar{w}^l}} 
 + \hbar  \frac{\partial}{\partial {\bar{w}^l}} \right]
\nonumber \\
&=  \left[ \sum_{n=0}^{\infty} \hbar^n a^b_{n, \vec{\bar{i}}}(f) 
{\left(\frac{\partial \bar{w}}{\partial {\bar{z}}}
 \right)^{\vec{\bar{i}}}}_{\vec{\bar{j}}} D^{\vec{\bar{j}}}_a~ , 
 ~ \frac{\partial \bar{z}^k}{\partial {\bar{w}^l}} 
 \left( \frac{\partial \Phi_a}{\partial {\bar{z}^k}}
 + \frac{\partial \bar{\phi}}{\partial {\bar{z}^k}}
 + \hbar  \frac{\partial}{\partial {\bar{z}^k}} \right) \right]
 \nonumber \\
 &= \frac{\partial \bar{z}^k}{\partial {\bar{w}^l}}
 \left[ \sum_{n=0}^{\infty} \hbar^n a^b_{n, \vec{\bar{i}}}(f) 
 {\left(\frac{\partial \bar{w}}{\partial {\bar{z}}}
 \right)^{\vec{\bar{i}}}}_{\vec{\bar{j}}} 
 D^{\vec{\bar{j}}}_a ~ , 
 \frac{\partial \Phi_a}{\partial {\bar{z}^k}} 
 + \hbar  \frac{\partial}{\partial {\bar{z}^k}} \right]
 =0,
\end{align}
and thus we obtain 
\begin{align}
L_{a,f} = \sum_{n=0}^{\infty} \hbar^n a^b_{n, \vec{\bar{i}}}(f) 
{\left(\frac{\partial \bar{w}}{\partial {\bar{z}}}
 \right)^{\vec{\bar{i}}}}_{\vec{\bar{j}}} 
 D^{\vec{\bar{j}}}_a = L_{b,f} \label{a_b_2}
\end{align}
which satisfies the condition $[L_{a,f} , R_{a, \partial_{\bar{i}} \Phi_a}]=0$.
\end{proof}
(\ref{a_b_2}) means that
\begin{align}
 a^b_{n, \vec{\bar{j}}}(f) 
 {\left(\frac{\partial \bar{w}}{\partial {\bar{z}}}
 \right)^{\vec{\bar{j}}}}_{\vec{\bar{i}}} =
 a^a_{n, \vec{\bar{i}}}(f) ,
\end{align}
in other words, $a^b_{n, \alpha}(f)$ transforms as a tensor.

As a next step, we consider the transition function between twisted Fock
representations. 
{}From Lemma \ref{lem-1},
we can choose $\Phi_a (z,\bar{z})$ and $\Phi_b (w,\bar{w})$ such that
\begin{align}
\Phi_a (0,\bar{z})=\Phi_a (z, 0)=0, ~\   
\Phi_b (0 ,\bar{w})= \Phi_b (w, 0)=0.
\end{align}
Using these K\"ahler potentials, $|\vec{0} \rangle_p {}_p\langle \vec{0}|$
is defined as 
\begin{align*}
 |\vec{0} \rangle_p {}_p\langle \vec{0}| = e^{-\Phi_p /\hbar}, ~\ ~ 
(p=a,b ) ,
 \end{align*} 
and 
 $|\vec{m} \rangle_p \underline{ {}_p\langle \vec{n}|}$ are defined by
\begin{align*}
 |\vec{m} \rangle_a  \underline{{}_a\langle \vec{n}|}
 &= \frac{1}{\sqrt{\vec{m}! \vec{n}!}}
 (z)_*^{\vec{m}} * e^{-\Phi_a/\hbar}
 * \left(\frac{1}{\hbar} \partial \Phi_a\right)_*^{\vec{n}}, 
\notag \\
|\vec{m} \rangle_b  \underline{{}_b\langle \vec{n}|}
 &= \frac{1}{\sqrt{\vec{m}! \vec{n}!}}
 (w)_*^{\vec{m}} * e^{-\Phi_b/\hbar}
 * \left(\frac{1}{\hbar} \partial \Phi_b\right)_*^{\vec{n}}.
\end{align*}

Let us consider the case that on the overlap $U_a \cap U_b$ 
the coordinate transition function $w(z)$, $\bar{w}(\bar{z})$,
and the functions $\exp (\phi(w)/\hbar )$ and $\exp (\bar{\phi}(\bar{w})/\hbar )$
are given by analytic functions.
Then the products $(w(z))^{\vec{\alpha}} \exp -(\phi(w)/\hbar ) $
and $(\bar{w}(\bar{z}))^{\vec{\alpha}}\exp -(\bar{\phi}(\bar{w})/\hbar ) $
are also analytic functions;
\begin{align} \label{a_b_3}
(w(z))^{\vec{\alpha}} e^{-\phi(w)/\hbar}
&= \sum_{\vec{\beta}} C^{\vec{\alpha}}_{\vec{\beta}} z^{\vec{\beta}},  
\notag \\
(\bar{w}(\bar{z}))^{\vec{\alpha}}e^{-\bar{\phi}(\bar{w})/\hbar} 
&= \sum_{\vec{\beta}} \bar{C}^{\vec{\alpha}}_{\vec{\beta}} \bar{z}^{\vec{\beta}}.
\end{align}

By using (\ref{vac-dPhi}), the bases are also written as
\begin{align}
 |\vec{m} \rangle_a \underline{{}_a\langle \vec{n}|}
 &= \sqrt{\frac{\vec{n}!}{\vec{m}!}} 
 \sum_{\vec{k}} H_{\vec{n}, \vec{k}}^a 
 (z)^{\vec{m}} (\bar{z})^{\vec{k}} e^{-\Phi_a/\hbar}, \notag \\
 |\vec{m} \rangle_b \underline{{}_b\langle \vec{n}|}
 &= \sqrt{\frac{\vec{n}!}{\vec{m}!}} 
 \sum_{\vec{k}} H_{\vec{n}, \vec{k}}^b 
 (w)^{\vec{m}} (\bar{w})^{\vec{k}} e^{-\Phi_b/\hbar}. 
\end{align} 
{} From the (\ref{a_b_3}) 
\begin{align}
|\vec{m} \rangle_b \underline{{}_b\langle \vec{n}|}
&= \sqrt{\frac{\vec{n}!}{\vec{m}!}} 
 \sum_{\vec{k}} H_{\vec{n}, \vec{k}}^b 
 (\sum_{\vec{\alpha}} C^{\vec{m}}_{\vec\alpha} z^{\vec{\alpha}} ) 
(\sum_{\vec{\beta}} \bar{C}^{\vec{k}}_{\vec\beta} \bar{z}^{\vec{\beta}} ) e^{-\Phi_a/\hbar} .
\end{align}

Finally, we obtain transformation between the bases,
\begin{align}
T^{ab} : F_{U_a} \rightarrow F_{U_b},
\end{align}
as
\begin{align}
|\vec{m} \rangle_b  \underline{{}_b\langle \vec{n}|}
= \sum_{\vec{i},\vec{j}}
T_{\vec{m}\vec{n}}^{ba, \vec{i}\vec{j}} |\vec{i} \rangle_a \underline{{}_a\langle \vec{j}|},
\end{align}
where 
\begin{align}
T_{\vec{m}\vec{n}}^{ba, \vec{i}\vec{j}}
=
\sqrt{\frac{\vec{n}!}{\vec{m}!}} 
\sqrt{\frac{\vec{i}!}{\vec{j}!}} 
 \sum_{\vec{k}} H_{\vec{n}, \vec{k}}^b 
 ( C^{\vec{m}}_{\vec{i}} ) 
(\sum_{\vec{\beta}} \bar{C}^{\vec{k}}_{\vec\beta} H^{a -1}_{\vec{\beta} \vec{j}} ) .
\end{align}

Using this transformation, the twisted Fock representation is extended to $M$.
We call it the twisted Fock representation of $M$.

\section{Trace operation} \label{sect5}

A trace operation to the Fock algebra is studied in this section.
A trace density $\mu$ of noncommutative manifolds $(M, *)$ is defined as
a density such that
\begin{align}\label{cycl}
\int_M \mu f * g = \int_M \mu g * f
\end{align}
for any functions $f \in C^{\infty}(M)$ and $g \in C_0^{\infty}(M)$, where 
$C_0^{\infty}$ is used as a set of compactly supported smooth functions.
Let ${\rm Tr}_M : C^{\infty}(M) \rightarrow {\mathbb C}\cup \infty$ be an integration 
with this trace density $ \mu $:
\begin{align}
{\rm Tr}_M  f := \int_M \mu f ~.
\end{align}
The existence of the trace density for the noncommutative K\"ahler manifolds
with the deformation quantization with separation of variables are 
guaranteed by the study in
\cite{Karabegov:1998hm}.
Note that elements of the basis of the twisted Fock algebras 
$|\vec{n} \rangle \underline{ \langle \vec{m}| }$ are not necessarily
compactly supported functions on an each local coordinate open set, in general.

Let $\{ (U_p , \phi_p )\}$ is an atlas of a K\"ahler manifold $M$ and we
use the notation $V_p = \phi_p (U_p) \subset {\mathbb R}^{2N}$.
For a bounded function given in the form of $f*g$, an integral over 
$\displaystyle M=\cup_p U_p$
 of $f*g$ with the trace density $\mu_{p}$ on $V_p$ is defined such as
\begin{align}
\int_M f*g \mu = \sum_{p} \int_{V_p} \rho_p f*g \mu_p ,
\end{align}
where $\rho_p$ is a partition of unity.
Note that 
\begin{align}
\sum_{p} \int_{V_p} \rho_p (f*g) \mu_p = 
\sum_{p} \int_{V_p} (\rho_p f ) *g \mu_p,
\label{cyclic_func}
\end{align}
because $\rho_p$ is an element of a partition of unity.
Therefore for a compact closed K\"ahler manifold $M$,
cyclic symmetries (\ref{cycl}) hold for 
arbitrary $f , g \in C^{\infty}(M)$.
In the following of this section, we fix a partition of unity on $M$.

Let us define a linear operation for the Fock algebra as follows.
\begin{defn}\label{trace}
Let $(U_p , \phi_p)$ be a chart of $M$. 
The local linear operation ${\rm Sp}_p $ on each $U_p$ is defined as 
a linear map from $F_{U_p}$ to ${\mathbb C}\cup \infty$ such that
\begin{align}
{\rm Sp}_p A * B = {\rm Sp}_p B* A
\end{align}
for arbitrary elements $A$ and $B$ of twisted Fock representation, and 
\begin{align}
{\rm Sp}_p | \vec{0} \rangle_p {}_p \langle \vec{0} | 
= c_p.
\end{align}
Here $c_p$ is a constant depending on $p$.
\end{defn}
Note that 
for the case that the considering K\"ahler manifold $M$ is ${\mathbb C}^n$,
the ${\rm Sp}_{{\mathbb C}^n}$ is equal to the trace operation 
${\rm Tr}$ up to the $c_p$. (See Example 1. in Section \ref{sect6}.)

\begin{rem}
The cyclic symmetry of the ${\rm Sp}_p$ operation and the commutation relations
of the creation and annihilation operators determine the the ${\rm Sp}_p$ of
bases of the twisted Fock representation.
\begin{align}
{\rm Sp}_p | \vec{n} \rangle_p \underline{ {}_p \langle \vec{m} |} = c_p
\delta_{\vec{n} \vec{m}}
=c_p \delta_{n_1 m_1} \delta_{n_2 m_2} \cdots \delta_{n_N m_N} . \label{tr_2}
\end{align}
\end{rem}
Indeed, the fact that 
the trace of commutator 
$|\vec{n} \rangle_p \underline{{}_p \langle \vec{m} |}$ 
and number operator $N_i = a_i^{\dagger} \underline{a}_i $ is zero,
\begin{align}
0 = {\rm Sp}_p 
 [N_i ~, ~ |\vec{n} \rangle_p \underline{{}_p \langle \vec{m} |} ]_*
 = (n_i - m_i ) {\rm Tr_p}  
 |\vec{n} \rangle_p \underline{{}_p \langle \vec{m} |},
 ~~~~ ( i= 1, \cdots , N ), 
\end{align}
implies
\begin{align}
 {\rm Sp}_p  |\vec{n} \rangle_p \underline{{}_p \langle \vec{m} |} 
 = c \delta_{\vec{n} \vec{m}},
\end{align}
where $c$ is some constant.
Furthermore, because of
$\displaystyle |\vec{n} \rangle_p = \frac{1}{\sqrt n_i} a_i^{\dagger} 
|\vec{n} - \vec{e}_i \rangle_p $,
\begin{align}
{\rm Sp}_p  |\vec{n} \rangle_p \underline{{}_p \langle \vec{n} |}
&= \frac{1}{n_i} {\rm Sp}_p ~a_i^{\dagger} 
 |\vec{n} - \vec{e}_i \rangle_p 
 \underline{{}_p \langle \vec{n} - \vec{e}_i |}\underline{a}_i
\notag \\
&=  \frac{1}{n_i} {\rm Sp}_p (N_i +1) |\vec{n} - \vec{e}_i \rangle_p 
\underline{{}_p \langle \vec{n} - \vec{e}_i |} \notag \\
&={\rm Sp}_p  |\vec{n} - \vec{e}_i \rangle_p 
\underline{{}_p \langle \vec{n} - \vec{e}_i |}\\
&\vdots \notag \\
&= {\rm Sp}_p  |\vec{0} \rangle_p \underline{{}_p \langle \vec{0} |}= c_p .
\end{align}
These results show the equation (\ref{tr_2}). 
Note that (\ref{tr_2}) are derived by using only commutation relations 
between the $a_i^{\dagger}$'s and $\underline{a}_i$'s and 
cyclic symmetry ${\rm Sp}_p A*B = {\rm Sp}_p B*A$.

Let us consider the relation between the trace operations and 
integration.
\begin{prop} \label{local trace}
If cyclic symmetries
\begin{align} \label{sym}
\int_U \mu_U 
 ~(a^\dagger)_*^{\vec{m}} * | \vec{0} \rangle \langle \vec{0} |
 * (\underline{a})_*^{\vec{n}} 
= \int_U \mu_U 
 ~|\vec{0} \rangle \langle \vec{0}|
 * (\underline{a})_*^{\vec{n}} *
(a^\dagger)_*^{\vec{m}} 
\end{align}
are satisfies for arbitrary $\vec{n} , \vec{m} $,
and 
$ c_0 := \int_U \mu_U | \vec{0} \rangle \langle \vec{0}| 
 = \int_U \mu_U \exp(-\Phi/\hbar)$ 
is finite,
then 
\begin{align}
\int_U \mu_U ~|\vec{m} \rangle \underline{\langle \vec{n}|}
= c_0 \delta_{\vec{m}, \vec{n}} .
\end{align}
\end{prop}

\begin{proof}
\begin{align}
 \int_U \mu_U ~|\vec{m} \rangle \underline{\langle \vec{n}|}
 &= \frac{1}{\sqrt{\vec{m}! \vec{n}!}} \int_U \mu_U 
 ~(a^\dagger)_*^{\vec{m}} * |\vec{0} \rangle \langle \vec{0}|
 * (\underline{a})_*^{\vec{n}} \nonumber \\
 &= \delta_{\vec{m}, \vec{n}} \int_U \mu_U |\vec{0} \rangle \langle \vec{0}|
 \nonumber \\
 &= c_0 \delta_{\vec{m}, \vec{n}} .
\end{align}
\end{proof}
For example, for ${\mathbb C}^N$ and ${\mathbb C}H^N$ we
can chose $U$ and $\Phi$ to satisfy the conditions 
in Proposition \ref{local trace}.
In such cases, ${\rm Sp}$ operation on $U$ 
\begin{align}
{\rm Sp}_U | \vec{m} \rangle \langle \vec{n} | := \delta_{\vec{m} \vec{n}} c_U
\end{align}
is expressed by the above integration, i.e.
\begin{align} \label{Sp_Tr}
{\rm Sp}_U | \vec{m} \rangle \langle \vec{n} | 
=\frac{c_U}{c_0} \int_U \mu_g ~|\vec{m} \rangle \underline{\langle \vec{n}|} .
\end{align}
 Here $c_U$ is a some constant.
Then the results of the trace operation of the twisted Fock algebra 
are given by easy algebraic calculation of ${\rm Sp}$.

For general K\"ahler manifolds 
there might not exist the open covering $M= \cup_p U_p$ 
such that each $U_p$ satisfies conditions (\ref{sym}) in this proposition.
In such case, (\ref{Sp_Tr}) does not work.
In addition, we have to introduce a partition of unity to describe the 
integration over a whole manifold, in general case.
But it is unknown whether
the partition of unity belongs to the twisted Fock algebra or not for the 
general case.
Therefore, we can not naively compare $\rm Sp$ operations with $\rm Tr$.
If globally defined twisted Fock algebra exists on a  K\"ahler manifolds,
then integral is evaluated by algebraic process.
This problem is discussed for some cases in Section \ref{sect6}.


Then how can we estimate the trace operation by using ${\rm Sp}$?
The ${\rm Sp}$ is related with integral over $V_p$ under some conditions.
\begin{prop}
Let $(U_p , \phi_p)$ be a chart satisfying 
$\phi_p (U_p) = V_p \subset {\mathbb R}^{2N}$ and $\rho_p $ be a 
a partition of unity corresponding to $U_p$. 
Consider that volumes of every overlapping domains 
between $U_p$ and $U_q ~ (q\neq p)$ 
are bounded by arbitrary positive number $\epsilon$, 
and values of all commutators between $\rho_p$ and $a_i^{\dagger}$ or
 $\underbar{a}_j$ are bounded by $1/\epsilon^{1-\delta_p}$ where
 $\delta_p$ is a real number.  
Then the integration is related with ${\rm Sp}$
\begin{align} \label{apporo}
{\rm Sp}_p | \vec{n} \rangle_p \underline{ {}_p \langle \vec{m} |}
=  \int_{V_p}  \rho_p | \vec{n} \rangle_p \underline{ {}_p \langle \vec{m} |} \mu_g +O(\epsilon^{{\delta}_p}), \
\end{align}
with the $c_p$ is a constant given by
\begin{align}
c_p = \int_{V_p} e^{-\Phi_p / \hbar} \mu_g .
\end{align}
\end{prop}

\begin{proof}
At the above remark, we made sure that this ${\rm Sp}_p$ operation 
is determined by the cyclic symmetry of the trace and the algebraic relations of $a_i^{\dagger}$'s and $\underline{a}_i$'s.
By the partition of unity $\rho_p$, the operation $\int_{V_p} \mu_g$
has cyclic symmetry for any elements in the twisted Fock algebra.
The only problem is $\rho_p$ does not commute with 
$a_i^{\dagger}$'s and $\underline{a}_i$'s on the over lapping region.
{}From the condition  that volumes of every overlapping domains 
between $U_p$ and $U_q ~ (q\neq p)$ 
are bounded by an arbitrary positive number $\epsilon$, 
and values of all commutators between $\rho_p$ and $a_i^{\dagger}$ or 
$\underbar{a}_j$ are bounded by $1/\epsilon^{1-\delta}$, 
(\ref{apporo}) is trivial.
\end{proof}



\section{Examples} \label{sect6}
In this section, some examples of the Fock representations are given.\\

\noindent
{\bf Example 1}: Fock representation of ${\mathbb C}^N$

The first example is ${\mathbb C}^N$.
The K\"ahler potential is given by 
\begin{align}
 \Phi_{{\mathbb C}^N} = \sum_{i=1}^N |z^i |^2 . 
\label{euclid_kahler_potential}
\end{align} 
By the process given in Section \ref{sect2}, the star product
is easily obtained as
\begin{align}
f*g &= \sum_{n=0}^{\infty} \frac{\hbar^n}{n!}
\delta^{k_1 l_1} \cdots \delta^{k_n l_n}
(\partial_{\bar{k}_1} \cdots \partial_{\bar{k}_n} f )
(\partial_{{l}_1} \cdots \partial_{{l}_n} g ) . \label{star_euclid}
\notag
\end{align}
This star product was given in \cite{Karabegov1996}.
We put 
\begin{align}
a_i^{\dagger} = z^i, ~~~
\underline{a}_i = \frac{1}{\hbar}\bar{z}^i, ~~~
a_i = \bar{z}^i, ~~~
\underline{a}_i^{\dagger} = \frac{1}{\hbar}{z}^i .
\end{align}
Then they satisfy the commutation relations:
\begin{align}
[\underline{a}_i ,~ a_j^{\dagger} ]_*= \delta_{ij},~~~
 [{a}_i ,~ \underline{a}_j^{\dagger} ]_* = \hbar \delta_{ij} 
\end{align}
and the others are zero.
Since in this case the operators with the underline are essentially equal
to those without the underline, we omit the underline of the bra vectors. 
The basis of the twisted Fock algebra is given by
\begin{align}
| \vec{m} \rangle  \langle \vec{n} | 
= \frac{1}{\hbar^{|\vec{n}|}\sqrt{\vec{m}! \vec{n}!}}
(z)^{\vec{m}}e^{-\Phi/\hbar} (\bar{z})^{\vec{n}} . \label{base_euclid}
\end{align}
These are defined globally, so the trace operations 
for the twisted Fock algebras given by
integral over ${\mathbb C}^N$ is equal to the ${\rm Sp}$ operation with
$c_{{\mathbb C}^N} =1 $:
\begin{align}
{\rm Tr}_{{\mathbb C}^N}  | \vec{m} \rangle \langle \vec{n} | 
:= \frac{\hbar^N}{\pi^N}\int_{{\mathbb C}^N} dx^{2D}  
| \vec{m} \rangle \langle \vec{n} | 
={\rm Sp}_{{\mathbb C}^N}  | \vec{m} \rangle \langle \vec{n} | 
= \delta_{\vec{m} \vec{n}}.
\end{align}
These results coincide with well known facts for noncommutative
Euclidean spaces.\bigskip \\

\noindent {\bf Example 2} : 
Fock representation of a cylinder\\ 
The second example is a cylinder $C$.  Let us consider the cylinder as a
special case of ${\mathbb C}$ with an equivalence relation $z \sim z + 2\pi$.
We choose its open covering as $C = U_a \cup U_b$, where $\displaystyle
U_a= \{ z_a \in {\mathbb C} | -\frac{\pi}{2} < {\rm Re}~ z_a < \pi \}$
and $\displaystyle U_b= \{ z_b \in {\mathbb C} | \frac{\pi}{2} < {\rm
Re}~ z_b < 2\pi \}$.  Then there are two overlap regions. The first one
is $A= \{ z \in {\mathbb C} | \frac{\pi}{2} < {\rm Re} z < \pi \} \subset U_a
\cap U_b$, and the transition functions on it is given by the identity
$z_a = z_b $.  The other overlap is $B= \{ z \in {\mathbb C} | \frac{3
\pi}{2} < {\rm Re} z < 2\pi \} \subset U_b$. On $B$, the transition
function is given by $z_b = z_a + 2\pi$.  The K\"ahler potential is
defined by (\ref{euclid_kahler_potential}) for $N=1$ and the star
products on the cylinder is also given by (\ref{star_euclid}) on each of
the open subset $U_a$ and $U_b$.  The basis of the local twisted Fock
algebra $| \vec{m} \rangle \langle \vec{n} |$ on $U_a$ and $U_b$ are
also given as $(\ref{base_euclid})$.  However, we can not describe them
globally since they do not have translation invariance under $z
\rightarrow z +2\pi$.  Thus ${\rm Tr}_C$ can not be
represented by using ${\rm Sp}_a$ and ${\rm Sp}_b$. \bigskip \\

\noindent
{\bf Example 3} : Fock representation of noncommutative $\mathbb{C}P^N$\\
We give an explicit expression of the twisted Fock representation
of noncommutative of $\mathbb{C}P^N$.
In this case, the twisted Fock representation on an open set 
is essentially the same as the representation given in 
\cite{Sako:2012ws,Sako:2013noa,Maeda:2015bnb,Sako:2015yrr}.
(In a context of a Fuzzy $\mathbb{C}P^N$, which is a different approach to 
noncommutative $\mathbb{C}P^N$, the Fock representations are 
discussed in \cite{Alexanian,Alexanian:2001qj,CarowWatamura:2004ct}.)

Let denote $\zeta^a ~(a=0, 1, \dots, N)$ homogeneous coordinates and 
$\bigcup U_a ~(U_a = \{[\zeta^0: \zeta^1: \cdots : \zeta^N]\} 
| \zeta^a \neq 0)$ an open covering of $\mathbb{C}P^N$. We define
inhomogeneous coordinates on $U_a$ as 
\begin{equation}
 z_a^0 = \frac{\zeta^0}{\zeta^a}, ~\cdots,~
  z_a^{a-1} = \frac{\zeta^{a-1}}{\zeta^a}, ~
  z_a^{a+1} = \frac{\zeta^{a+1}}{\zeta^a}, ~\cdots,~
  z_a^N = \frac{\zeta^N}{\zeta^a}.
\end{equation}
We choose a K\"ahler potential on $U_a$ which satisfies the condition
(\ref{Phi_cond})
\begin{equation}
 \Phi_a = \ln (1 + |z_a|^2),
\end{equation}
where $|z_a|^2 = \sum_i |z_a^i|^2$.
A star product on $U_a$ is given as follows
 \cite{Sako:2012ws,Sako:2013noa}:
\begin{align}
 f*g &= \sum_{n=0}^\infty c_n (\hbar) 
 g_{j_1 \bar{k}_1} \cdots g_{j_n \bar{k}_n} 
 \left(D^{j_1} \cdots D^{j_n} f\right)
 D^{\bar{k}_1} \cdots D^{\bar{k}_n} g,
 \label{Lf-cov}
\end{align}
where
\begin{align}
 c_n (\hbar) &= \frac{\Gamma(1-n+1/\hbar)}{n! \Gamma(1+1/\hbar)},~~~~ 
D^{\bar i} = g^{{\bar i} j} \partial_j ,~~~~ 
D^{i} = g^{i {\bar j} } \partial_{\bar j}.
\end{align}

On $U_a$, creation and annihilation operators are given as
\begin{equation}
 a_{a,i}^\dagger = z_a^i, ~~~
  \underline{a_{a,i}} = \frac{1}{\hbar} \partial_i \Phi_a
  = \frac{\bar{z}_a^i}{\hbar{(1+|z_a|^2)}}, ~~~
  a_{a,i} = \bar{z}_a^i, ~~~
  \underline{a_{a,i}}^\dagger = \frac{1}{\hbar} \partial_{\bar{i}} \Phi_a
  = \frac{z_a^i}{\hbar{(1+|z_a|^2)}}.
\end{equation}
and a vacuum is
\begin{equation}
 |\vec{0} \rangle_a \underline{_a \langle \vec{0}|}
  = e^{-\Phi_a/\hbar} = (1 + |z_a|^2)^{-1/\hbar}.
\end{equation}
Bases of the Fock representation on $U_a$ are constructed as
\begin{align}
 |\vec{m} \rangle_a \underline{_a \langle \vec{n}|}
 &= \frac{1}{\sqrt{\vec{m}! \vec{n}!}}
 (a_a^\dagger)_*^{\vec{m}} * 
 |\vec{0} \rangle_a \underline{_a \langle \vec{0}|}
 * (\underline{a_a})_*^{\vec{n}} 
 \nonumber \\
 &= \frac{1}{\sqrt{\vec{m}! \vec{n}!} \hbar^{|n|}}
 (z_a)_*^{\vec{m}} * e^{-\Phi_a/\hbar}
 * (\partial \Phi_a)_*^{\vec{n}}.
\end{align}
By using (\ref{vacuum1}), (\ref{vacuum2}) and the following relation
which is shown in \cite{Sako:2012ws},
\begin{align}
 (\partial \Phi_a)_*^{\vec{n}}
 &= \frac{\hbar^{|n|} \Gamma(1/\hbar+1)}{\Gamma(1/\hbar-|n|+1)}
 (\partial \Phi_a)^{\vec{n}} \nonumber \\
 &= \frac{\hbar^{|n|} \Gamma(1/\hbar+1)}{\Gamma(1/\hbar-|n|+1)}
 \left(\frac{\bar{z}_a}{1+|z_a|^2}\right)^{\vec{n}},
\end{align}
the bases can be explicitly written as
\begin{equation}
 |\vec{m} \rangle_a \underline{_a \langle \vec{n}|}
  = \frac{\Gamma(1/\hbar+1)}{\sqrt{\vec{m}! \vec{n}!}
  \Gamma(1/\hbar-|n|+1)}
  (z_a)^{\vec{m}} (\bar{z}_a)^{\vec{n}} e^{-\Phi/\hbar}. \label{base_CP}
\end{equation}
By comparing this equation and (\ref{base_H_ver}), $H_{\vec{m}, \vec{n}}$ is
obtained as
\begin{equation}
 H_{\vec{m}, \vec{n}} = \delta_{\vec{m}, \vec{n}}
  \frac{\Gamma(1/\hbar+1)}{\vec{m}! \Gamma(1/\hbar-|m|+1)},
\end{equation}
and it is easily seen that this formally satisfies 
$e^{\Phi_a/\hbar} = \sum H_{\vec{m}, \vec{n}} 
(z_a)^{\vec{m}}(\bar{z})^{\vec{n}}$.

Let us consider transformations between the Fock representations on
$U_a$ and $U_b$ ($a < b$).
The transformations for the coordinates and the K\"ahler potential on
$U_a \bigcap U_b$ are
\begin{align}
 z_a^i &= \frac{z_b^i}{z_b^a}, ~~
 (i = 0, 1, \dots, a-1, a+1, \dots, b-1, b+1, \dots, N),
 \qquad
 z_a^b = \frac{1}{z_b^a}, \\
 \Phi_a &= \Phi_b - \ln z_b^a - \ln \bar{z}_b^a. 
\end{align} 
Thus, $|\vec{m} \rangle_a \underline{_a \langle\vec{n}|}$ is written on 
$U_a \bigcap U_b$ as
\begin{align}
 |\vec{m} \rangle_a \underline{_a \langle\vec{n}|}
 &= \frac{\Gamma(1/\hbar+1)}{\sqrt{\vec{m}!\vec{n}!}
 \Gamma(1/\hbar-|n|+1)} e^{-\Phi_b/\hbar} 
 \nonumber \\
 & ~~
 \times (z_b^0)^{m_0} \cdots
 (z_b^{a-1})^{m_{a-1}}(z_b^a)^{1/\hbar-|m|}
 (z_b^{a+1})^{m_{a+1}} \cdots
 (z_b^{b-1})^{m_{b-1}}(z_b^{b+1})^{m_{b+1}}
 \cdots (z_b^N)^{m_N}
 \nonumber \\
 & ~~
 \times (\bar{z}_b^0)^{n_0} \cdots
 (\bar{z}_b^{a-1})^{n_{a-1}}(\bar{z}_b^a)^{1/\hbar-|n|}
 (\bar{z}_b^{a+1})^{n_{a+1}} \cdots
 (\bar{z}_b^{b-1})^{n_{b-1}}(\bar{z}_b^{b+1})^{n_{b+1}}
 \cdots (\bar{z}_b^N)^{n_N},
\end{align}
where
\begin{align}
 \vec{m} &= (m_0, \dots, m_{a-1}, m_{a+1}, \dots, m_N), \\
 \vec{n} &= (n_0, \dots, n_{a-1}, n_{a+1}, \dots, n_N).
\end{align}
We should treat
$(z_b^a)^{1/\hbar-|m|}$ and $(\bar{z}_b^a)^{1/\hbar-|n|}$
carefully, because if they are not monomials 
some trick is needed to express them as the twisted Fock representation.
We here make comments about the trick briefly.
{}From the expression of the basis (\ref{base_CP}), 
a function $f(z,\bar{z}) e^{-\Phi / \hbar}$ 
is expressed as the Twisted Fock algebra
when $f(z, \bar{z})$ is given as a Taylor expansion in $z$ and $\bar{z}$.
For simplicity, we consider the one dimensional case.
When a non-monomial function $z^{q}$ of some complex coordinate $z$ 
with a nonpositive integer $q$ is given, $z^{q}$ should be Taylor expanded 
around the some non-zero point $p \in {\mathbb C}$ to express it as a twisted 
Fock algebra:
\begin{align}
z^q = p^q + qp^{q-1} (z-p)+ \frac{q(q-1)}{2}p^{q-2} (z-p)^2 + \cdots \ .
\end{align}
In the case that the radius of convergence of this expansion is
not enough to cover the whole of $U_b$, 
we have to divide $U_b$ into smaller ones, 
$U_b = \cup_{b_i} U_{b_i}$ and choose proper points for the Taylor expansions
on each $U_{b_i}$, to make each expansions converge.
For the higher dimensional ${\mathbb C}P^N$ 
we can use a similar procedure to the one dimensional case,
and the twisted Fock algebra for ${\mathbb C}P^N$ is derived.

To avoid such kind of complications concerning 
$(z_b^a)^{1/\hbar-|m|}$ and $(\bar{z}_b^a)^{1/\hbar-|n|}$, 
we can introduce a slightly different representation
from the above twisted Fock representation of ${\mathbb C}P^N$.
Let us consider the case that the noncommutative parameter is the
following value,  
\begin{align}
 & 1/\hbar = L \in \mathbb{Z}, ~~L \geq 0, 
 \label{L} 
\end{align}
Then, we define $F_a^L$ on $U_a$ as a subspace
of a local twisted Fock algebra $F_{U_a}$, 
\begin{align}
 F_a^L = \{ \sum_{\vec{m}, \vec{n}} A_{\vec{m} \vec{n}} 
 |\vec{m} \rangle_a \underline{_a \langle \vec{n}|} ~|~ 
 A_{\vec{m} \vec{n}} \in \mathbb{C}, ~|m| \leq L, ~|n| \leq L\} .
\end{align} 
The bases on $U_a$ are related to 
those on $U_b$ as,
\begin{equation}
\sqrt{\frac{(L-|n|)!}{(L-|m|)!}} 
|\vec{m} \rangle_a \underline{_a \langle\vec{n}|}
= \sqrt{\frac{(L-|n'|)!}{(L-|m'|)!}}
|\vec{m'} \rangle_b \underline{_b \langle\vec{n'}|}, \label{trans_finite}
\end{equation}
where
\begin{align}
 \vec{m'} &= (m_0, \cdots, m_{a-1}, L-|m|, m_{a+1}, \cdots, 
 m_{b-1}, m_{b+1}, \cdots, m_N), \\
 \vec{n'} &= (n_0, \cdots, n_{a-1}, L-|n|, n_{a+1}, \cdots, 
 n_{b-1}, n_{b+1}, \cdots, n_N).
\end{align}
Using the expression of (\ref{trans_finite}), we can define 
$|\vec{m} \rangle_a \underline{_a \langle\vec{n}|}$ on the whole of $U_b$.
Therefore, the operators in $F_a^L$ can be extended to the whole of
$\mathbb{C}P^N$ by using the relation like (\ref{trans_finite}).

Under the condition (\ref{L}), 
the creation and annihilation operators on $F_{U_a}$ is changed from
the definition (\ref{creation_anihilation}).
Similarly to (\ref{ca-op1}) and (\ref{ca-op2}), let us define 
a creation operator ${a^L_{a,i}}^\dagger$ and
an annihilation operator $\underline{a}^L_{a,i}$ 
restricted on $F_a^L$ by
\begin{align}
 {a^L_{a,i}}^\dagger &= \sum_{0 \leq |n| \leq L-1} \sqrt{n_i +1} 
 |\vec{n} + \vec{e}_i \rangle_a \underline{_a \langle \vec{n} |}
 = z_a^i \left[ 1- \left(\frac{|z_a|^2}{1+|z_a|^2} \right)^L \right], \\
 \underline{a}^L_{a,i} &= \sum_{0 \leq |n| \leq L-1} \sqrt{n_i +1} 
 |\vec{n}  \rangle_a \underline{_a \langle \vec{n} + \vec{e}_i |}
 = L \frac{\bar{z}_a^i}{1+|z_a|^2}.
\end{align}
By the restriction on $F^L_{a}$, ${a^L_{a,i}}^\dagger$ is shifted from
$z_a^i$. These operators satisfy the following commutation relation,
\begin{align}
 [\underline{a}^L_{a,i}, ~{a^L_{a,j}}^\dagger]
 &= \delta_{ij} \left(
 \sum_{0 \leq |n| \leq L} |\vec{n} \rangle_a \underline{_a \langle \vec{n} |}
 - \sum_{|n|=L} (n_i +1) |\vec{n} \rangle_a \underline{_a \langle \vec{n} |}
 \right) \nonumber \\
 &= \delta_{ij}
 - \delta_{ij} \left(\frac{|z_a|^2}{1+|z_a|^2}\right)^L
 \left(1 + L\frac{|z_a^i|^2}{|z_a|^2}\right).
\end{align}
\bigskip


\noindent
{\bf Example 4} : Fock representation of noncommutative $\mathbb{C}H^N$\\
Here, we give an explicit expression of the Fock representation
of noncommutative of $\mathbb{C}H^N$ \cite{Sako:2012ws,Sako:2013noa}.

We choose a K\"ahler potential satisfies the condition
(\ref{Phi_cond})
\begin{equation}
 \Phi = - \ln (1 - |z|^2),
\end{equation}
where $|z|^2 = \sum_i^N |z^i|^2$.
A star product  is given as follows
 \cite{Sako:2012ws,Sako:2013noa}:
\begin{align}
 f*g &= \sum_{n=0}^\infty c_n (\hbar) 
 g_{j_1 \bar{k}_1} \cdots g_{j_n \bar{k}_n} 
 \left(D^{j_1} \cdots D^{j_n} f\right)
 D^{\bar{k}_1} \cdots D^{\bar{k}_n} g,
\end{align}
where
\begin{align}
 c_n (\hbar) &= \frac{\Gamma(1/\hbar)}{n! \Gamma(n+1/\hbar)},~~~~ 
D^{\bar i} = g^{{\bar i} j} \partial_j ,~~~~ 
D^{i} = g^{i {\bar j} } \partial_{\bar j}.
\end{align}

The creation and annihilation operators are given as
\begin{equation}
 a_{i}^\dagger = z^i, ~~~
  \underline{a_{i}} = \frac{1}{\hbar} \partial_i \Phi
  = \frac{\bar{z}^i}{\hbar{(1-|z|^2)}}, ~~~
  a_{i} = \bar{z}^i, ~~~
  \underline{a_{i}}^\dagger = \frac{1}{\hbar} \partial_{\bar{i}} \Phi
  = \frac{z^i}{\hbar{(1-|z|^2)}}.
\end{equation}
and a vacuum is
\begin{equation}
 |\vec{0} \rangle \langle \vec{0}|
  = e^{-\Phi/\hbar}= (1-|z|^2)^{1/\hbar} .
\end{equation}
Bases of the Fock representation on $\mathbb{C}H^N$ are constructed as
\begin{align}
 |\vec{m} \rangle \underline{ \langle \vec{n}|}
 &= \frac{1}{\sqrt{\vec{m}! \vec{n}!}}
 (a^\dagger)_*^{\vec{m}} * 
 |\vec{0} \rangle \underline{ \langle \vec{0}|}
 * (\underline{a})_*^{\vec{n}} 
 \nonumber \\
 &= \frac{1}{\sqrt{\vec{m}! \vec{n}!} \hbar^{|n|}}
 (z)_*^{\vec{m}} * e^{-\Phi/\hbar}
 * (\partial \Phi )_*^{\vec{n}}.
\end{align}
By using (\ref{vacuum1}), ($\ref{vacuum2}$) and the following relation
which is shown in \cite{Sako:2012ws},
\begin{align}
 (\partial \Phi)_*^{\vec{n}}
 &= \frac{(-\hbar)^{|n|} \Gamma(1/\hbar+|n|)}{\Gamma(1/\hbar)}
 \left(\frac{\bar{z}}{1-|z|^2}\right)^{\vec{n}},
\end{align}
the bases can be explicitly written as
\begin{equation}
 |\vec{m} \rangle \underline{ \langle \vec{n}|}
  = \frac{(-1)^{|n|}\Gamma(1/\hbar+|n|)}{\sqrt{\vec{m}! \vec{n}!}
  \Gamma(1/\hbar )}
  (z)^{\vec{m}} (\bar{z})^{\vec{n}} (1-|z|^2)^{1/\hbar}.
\end{equation}
These are defined globally.
For ${\mathbb C}H^N$, trace density is given by the usual Riemannian 
volume form
\begin{align}
\mu_g = \frac{1}{(1-|z|^2)^{N+1}}.
\end{align}
Therefore,
\begin{align}
c_0 &= \int_{{\mathbb C}H^N} dz^{2D} \mu_g |0 \rangle \langle 0 |
\notag \\
&=  \int_{{\mathbb C}H^N}  dz^{2D} \frac{1}{(1-|z|^2)^{\frac{1}{\hbar} -(N+1)}} \notag \\
&= \pi^N \frac{\Gamma (1/\hbar - N)}{\Gamma (1/ \hbar )},
\end{align}
and the trace is given by the integration
\begin{align}
{\rm Tr }_{{\mathbb C}H^N} | \vec{m} \rangle \underline{\langle \vec{n} |}
 = \frac{\Gamma (1/ \hbar )}{\pi^N \Gamma (1/\hbar - N)} 
\int_{{\mathbb C}H^N}   dz^{2D} \mu_g | \vec{m} \rangle \underline{\langle \vec{n} |} 
= \delta_{\vec{m} \vec{n}}.
\end{align}
\bigskip

At the end of this section,
we mention a special class of K\"ahler manifolds. 
The above examples, ${\mathbb C}^N$, cylinder, ${\mathbb C}P^N$ and 
${\mathbb C}H^N$, have K\"ahler potentials which depend only on 
the absolute values of complex coordinates:
\begin{align}
\Phi (z, \bar{z}) = 
\tilde{\Phi}( |z_1|, |z_2|, \cdots ,|z_N| ) . \label{special}
\end{align}
For this case, we obtain the usual Fock algebra by the following proposition.

\begin{prop}\label{special_phi}
When a K\"ahler potential is an analytic function and has the form of 
(\ref{special}), 
$|\vec{0} \rangle \langle \vec{0}|* (a)^{\vec{m}}_*$  and 
$(a^{\dagger})^{\vec{m}}_* * |\vec{0} \rangle \langle \vec{0} | $ are equal
to $|\vec{0} \rangle \langle \vec{0}| * (\underline{a} )^{\vec{m}}_*$ 
and $ (\underline{a}^{\dagger})^{\vec{m}}_* * |\vec{0} \rangle \langle \vec{0}|$
up to a constant, respectively. 
\end{prop}

\begin{proof}
{}From the identity 
$ L_{\partial_i \Phi } = \hbar e^{-\Phi /\hbar } \partial_i e^{\Phi /\hbar}$,
\begin{align}
&|\vec{0} \rangle \langle \vec{0}| * ( \partial_{1} \Phi )^{n_1}_* * \cdots * 
(\partial_{i_N} \Phi )^{n_N}_* 
\notag \\
&=\hbar^{|n|} |\vec{0} \rangle \langle \vec{0} |* \left( \bar{z}_{1}^{n_1} \cdots \bar{z}_{N}^{n_N}
e^{-\tilde{\Phi}/\hbar} \left( \frac{\partial}{\partial {|z_{1}|}} \right)^{n_1} \cdots 
\left( \frac{\partial}{\partial {|z_{N}|}}\right)^{n_N} e^{\tilde{\Phi}/\hbar} 
\right). \notag 
\end{align}
By using Lemma \ref{ssu},
$e^{-\Phi /\hbar } * f(z , \bar{z}) = e^{-\Phi /\hbar } * f(0 , \bar{z})$,
$$
e^{-\tilde{\Phi}/\hbar} \left( \frac{\partial}{\partial {|z_{1}|}} \right)^{n_1} \cdots 
\left( \frac{\partial}{\partial {|z_{N}|}}\right)^{n_N} e^{\tilde{\Phi}/\hbar} 
$$
in the above equation
can be replaced by a constant, which we here denote by $C(\vec{n})$. Then
\begin{align}
&|\vec{0} \rangle \langle \vec{0}| * ( \partial_{1} \Phi )^{n_1}_* * \cdots * 
(\partial_{i_N} \Phi )^{n_N}_* 
= \hbar^{|n|} C(\vec{n}) |\vec{0}\rangle \langle \vec{0}| * (\bar{z} )^{\vec{n}}_* .
\end{align}
\end{proof}

As a corollary we obtain the following.
\begin{cor}
When a K\"ahler potential is an analytic function and has the form of
 (\ref{special}),  
$| \vec{m} \rangle \underline{ \langle \vec{n} | }
= \hbar^{|n|}C (\vec{n}) | \vec{m} \rangle \langle \vec{n} | $.
Here 
$C(\vec{n}) = e^{-\tilde{\Phi}/\hbar} 
\left( \frac{\partial}{\partial {|z_{1}|}} \right)^{n_1} \cdots 
\left( \frac{\partial}{\partial {|z_{N}|}}\right)^{n_N} 
e^{\tilde{\Phi}/\hbar} |_{z=0}.
$
\end{cor}

This corollary is also shown by using the definition of
$H_{\vec{m}, \vec{n}}$ in (\ref{H}) 
without Proposition \ref{special_phi}.
{}For a K\"ahler potential $\tilde{\Phi}$, $H_{\vec{m}, \vec{n}}$ is
proportional to $\delta_{\vec{m}, \vec{n}}$,
\begin{align}
 e^{\tilde{\Phi}/\hbar} 
 &= \sum_{\vec{m}, \vec{n}} \frac{C(\vec{n})}{\vec{n}!} 
 \delta_{\vec{m}, \vec{n}}
 (z)^{\vec{m}} (\bar{z})^{\vec{n}}, \\
 H_{\vec{m}, \vec{n}} &= \frac{C(\vec{n})}{\vec{n}!} 
 \delta_{\vec{m}, \vec{n}}.
\end{align}
Here $C(\vec{n})$ is a constant and is given as
\begin{align}
 C(\vec{n}) =  
\left( \frac{\partial}{\partial {|z_{1}|}} \right)^{n_1} \cdots 
\left( \frac{\partial}{\partial {|z_{N}|}}\right)^{n_N} 
e^{\tilde{\Phi}/\hbar} |_{z=0}.
\end{align}
By using (\ref{vac-dPhi}), we find
\begin{align}
 |\vec{m} \rangle \underline{ \langle \vec{n}|} 
 = C(\vec{n}) |\vec{m} \rangle \langle \vec{n}| .
\end{align}

\section{Summary} \label{sect7}
Twisted Fock representations of general noncommutative K\"ahler
manifolds are constructed.  The noncommutative K\"ahler manifolds
studied in this article are given by deformation quantization with
separation of variables.  Using this type of deformation quantization,
the twisted Fock representation which constructed based on two sets of
creation and annihilation operators was introduced with the concrete
expressions of them on a local coordinate chart.  The corresponding
functions are given by the local complex coordinates, the K\"ahler
potentials and partial derivatives of them with respect to 
the coordinates.  
The dictionary to translate bases of the twisted Fock
representation into functions is given as table \ref{table:result}.  
They are defined
on a local coordinate chart, and they are extended by the transition
functions given in Section \ref{sect4}.  This extension is achieved by
essentially the result that the star products with separation of
variables have a trivial transition function.
We also gave examples of the twisted 
Fock representation of K\"ahler manifolds,
${\mathbb C}^N$, cylinder, ${\mathbb C}P^N$ and ${\mathbb C}H^N$.
The trace operation as an integration over a manifold 
is obtained by traces of matrix representations for the 
${\mathbb C}^N$ and ${\mathbb C}H^N$.
\bigskip

\noindent
{\bf Acknowledgments} \\
A.S. was supported in part by JSPS
KAKENHI Grant Number 16K05138.


\end{document}